\documentclass[11pt]{article}

\usepackage[margin=1in]{geometry}
\usepackage{CJK, hyperref, amsthm, amsfonts}

\newtheorem{lemma}{Lemma}
\newtheorem{theorem}{Theorem}

\theoremstyle{definition}
\newtheorem{definition}{Definition}

\theoremstyle{remark}
\newtheorem*{remark}{Remark}

\begin{document}

\begin{CJK*}{UTF8}{}

\title{Area law in one dimension:\\
Degenerate ground states and Renyi entanglement entropy}

\CJKfamily{gbsn}

\author{Yichen Huang (黄溢辰)\thanks{Supported by DARPA OLE.}\\
Department of Physics, University of California, Berkeley, Berkeley, California 94720, USA\\
yichenhuang@berkeley.edu}

\maketitle

\end{CJK*}

\begin{abstract}

An area law is proved for the Renyi entanglement entropy of possibly degenerate ground states in one-dimensional gapped quantum systems. Suppose in a chain of $n$ spins the ground states of a local Hamiltonian with energy gap $\epsilon$ are constant-fold degenerate. Then, the Renyi entanglement entropy $R_\alpha(0<\alpha<1)$ of any ground state across any cut is upper bounded by $\tilde O(\alpha^{-3}/\epsilon)$, and any ground state can be well approximated by a matrix product state of subpolynomial bond dimension $2^{\tilde O(\epsilon^{-1/4}\log^{3/4}n)}$.

\end{abstract}

\section{Introduction}

The area law states that for a large class of ``physical'' quantum many-body states the entanglement of a region scales as its boundary (area) \cite{ECP10}. This is in sharp contrast to the volume law for generic states \cite{HLW06}: the entanglement of a region scales as the number of sites in (i.e., the volume of) the region. In one dimension (1D), the area law is of particular interest for it characterizes the classical simulability of quantum systems. Specifically, bounded (or even logarithmic divergence of) Renyi entanglement entropy across all cuts implies efficient matrix product state (MPS) representations \cite{VC06}, which underlie the (heuristic) density matrix renormalization group (DMRG) algorithm \cite{Whi92, Whi93}. Since MPS can be efficiently contracted, the 1D local Hamiltonian problem with the restriction that the ground state satisfies area laws is in NP. Furthermore, a structural result (Lemma \ref{ls} and see also \cite{Has07, AKLV13}) from the proof of the area law for the ground state of 1D gapped Hamiltonians is an essential ingredient of the (provably) polynomial-time algorithm \cite{LVV14} for computing such states, establishing that the 1D gapped local Hamiltonian problem is in P. The area law is now a central topic in the emerging field of Hamiltonian complexity \cite{Osb12}.

We start with the definition of entanglement entropy.
\begin{definition} [Entanglement entropy]
The Renyi entanglement entropy $R_\alpha(0<\alpha<1)$ of a bipartite (pure) quantum state $\rho_{AB}$ is defined as
\begin{equation}
R_\alpha(\rho_A)=(1-\alpha)^{-1}\log\mathrm{tr}\rho_A^\alpha,
\end{equation}
where $\rho_A=\mathrm{tr}_B\rho_{AB}$ is the reduced density matrix. The von Neumann entanglement entropy is defined as
\begin{equation}
S(\rho_A)=-\mathrm{tr}(\rho_A\log\rho_A)=\lim_{\alpha\rightarrow1^-}R_\alpha(\rho_A).
\end{equation}
\end{definition}

Here are three arguments why Renyi entanglement entropy is more suitable than von Neumann entanglement entropy for formulating area laws, although the latter is the most popular entanglement measure (for pure states) in quantum information and condensed matter theory.\\
1 (conceptual, classical simulability). In 1D, (unlike bounded Renyi entanglement entropy) bounded von Neumann entanglement entropy across all cuts does not necessarily imply efficient MPS representations; see \cite{SWVC08} for a counterexample. Although slightly outside the scope of the present paper, related results are summarized in Table \ref{tab} (see also \cite{Har14}).\\
2 (conceptual, quantum computation). Quantum states with little von Neumann entanglement entropy across all cuts support universal quantum computation, while an analogous statement for Renyi entanglement entropy is expected to be false \cite{VdN13}.\\
3 (technical). An area law for Renyi entanglement entropy implies that for von Neumann entanglement entropy (Table \ref{tab}), as $R_\alpha$ is a monotonically decreasing function of $\alpha$.

\begin{table}
\caption{Relations between various conditions in 1D: unique ground state of a gapped local Hamiltonian (Gap), exponential decay of correlations (Exp), area law for Renyi entanglement entropy $R_\alpha,\forall\alpha$ (AL-$R_\alpha$), area law for von Neumann entanglement entropy (AL-$S$), efficient matrix product state representation (MPS). A check (cross) mark means that the item in the row implies (does not imply) the item in the column. The asterisk marks one contribution of the present paper. It is an open problem whether exponential decay of correlations implies area laws for Renyi entanglement entropy $R_\alpha,\forall\alpha$: Indeed, Theorem 4 in \cite{BH13} (or Theorem 1 in \cite{BH15}) may lead to divergence of $R_\alpha$ if $\alpha$ is small.}

\begin{tabular}{|c|c|c|c|c|}
\hline
&Exp&AL-$R_\alpha$&AL-$S$&MPS\\
\hline
Gap&\checkmark\cite{Has04}&\checkmark*&\checkmark\cite{Has07, AKLV13}&\checkmark\cite{Has07, AKLV13}\\
\hline
Exp&\checkmark&?&\checkmark\cite{BH13, BH15}&\checkmark\cite{BH13, BH15}\\
\hline
AL-$R_\alpha$&\textsf{X}&\checkmark&\checkmark&\checkmark\cite{VC06}\\
\hline
AL-$S$&\textsf{X}&\textsf{X}&\checkmark&\textsf{X}\cite{SWVC08}\\
\hline
\end{tabular}\label{tab}
\end{table}

Hastings first proved an area law for the ground state of 1D Hamiltonians with energy gap $\epsilon$: The von Neumann entanglement entropy across any cut is upper bounded by $2^{O(\epsilon^{-1})}$ \cite{Has07}, where the local dimension of each spin (denoted by ``$d$'' in qu\emph{d}its) is assumed to be an absolute constant. The Renyi entanglement entropy $R_\alpha$ for $\alpha_0<\alpha<1$ was also discussed, where $\alpha_0$ is $\epsilon$-dependent and $\lim_{\epsilon\rightarrow0^+}\alpha_0=1$. The bound on the von Neumann entanglement entropy was recently improved to $\tilde O(\epsilon^{-3/2})$ \cite{AKLV13} (see Section \ref{notes} for an explanation of this result), where $\tilde O(x):=O(x~\mathrm{poly}\log x)$ hides a polylogarithmic factor. These proofs of area laws assume a unique (nondegenerate) ground state.

Ground-state degeneracy is an important physical phenomenon often associated with symmetry breaking (e.g., the transverse field Ising chain) and/or topological order (e.g., the Haldane/AKLT chain with open boundary conditions). Since not all degenerate ground states of 1D gapped Hamiltonians have exponential decay of correlations, it may not be intuitively obvious to what extent they satisfy area laws.

In the present paper, an area law is proved for the Renyi entanglement entropy of possibly degenerate ground states in 1D gapped systems. Since in this context the standard bra-ket notation may be cumbersome, quantum states and their inner products are simply denoted by $\psi,\phi\ldots$ and $\langle\psi,\phi\rangle$, respectively, cf. $\||\psi\rangle-|\phi\rangle\|$ versus $\|\psi-\phi\|$. Suppose in a chain of $n$ spins the ground states are constant-fold degenerate.
\begin{theorem} \label{t1}
(a) The Renyi entanglement entropy $R_\alpha(0<\alpha<1)$ of any ground state across any cut is upper bounded by $\tilde O(\alpha^{-3}/\epsilon)$;\\
(b) Any ground state $\psi$ can be approximated by an MPS $\phi$ of subpolynomial bond dimension $2^{\tilde O(\epsilon^{-1/4}\log^{3/4}n)}$ such that $|\langle\psi,\phi\rangle|>1-1/\mathrm{poly}(n)$.
\end{theorem}

\begin{remark}
The proof of Theorem \ref{t1} assumes constant-fold exact ground-state degeneracy and open boundary conditions (with one cut). It should be clear that a minor modification of the proof leads to the same results in the presence of an exponentially small $2^{-\Omega(n)}$ splitting of the ground-state degeneracy (as is typically observed in physical systems) and works for periodic boundary conditions (with two cuts). However, it is an open problem to what extent degenerate ground states satisfy area laws if the degeneracy grows with the system size. Theorem \ref{t1}(b) is a theoretical justification of the practical success of DMRG as a (heuristic) variational algorithm over MPS to compute the ground-state space in 1D gapped systems with ground-state degeneracy, and paves the way for a (provably) polynomial-time algorithm to compute the ground-state space. As an important immediate corollary of Theorem \ref{t1}(a), the von Neumann entanglement entropy of a unique ground state is upper bounded by $\tilde O(\epsilon^{-1})$, which even improves the result of \cite{AKLV13} and may possibly be tight up to a polylogarithmic factor. An example with the von Neumann entanglement entropy $S=\tilde\Omega(\epsilon^{-1/4})$ was constructed in \cite{GH10}; see also \cite{Ira10} for a translationally invariant construction with $S=\Omega(\epsilon^{-1/12})$.
\end{remark}

We loosely follow the approach in \cite{AKLV13} with additional technical ingredients. Approximate ground-space projection (AGSP) \cite{ALV12} is a tool for bounding the decay of Schmidt coefficients: An ``efficient'' family of AGSP imply an area law. Section \ref{s3} is devoted to perturbation theory, which is necessary to improve the efficiency of AGSP. As a technical contribution, the analysis in Section 6 of \cite{AKLV13} is improved (and simplified), resulting in a tightened upper bound $\tilde O(\epsilon^{-1})$ (versus $\tilde O(\epsilon^{-3/2})$ given in \cite{AKLV13}) on the (von Neumann) entanglement entropy. Although the perturbation theory is developed in 1D, generalizations to higher dimensions may be straightforward but are not presented in the present paper. In Section \ref{s2}, a family of AGSP are constructed in 1D systems with nearly degenerate ground states. Although the ground-state degeneracy of the original Hamiltonian is assumed to be exact, perturbations may lead to an exponentially small splitting of the degeneracy. Then, ``fine tunning'' using Lagrange interpolation polynomials appears necessary to repair this splitting at the level of AGSP. In Section \ref{s4}, an area law is derived from AGSP for any ground state by constructing a sequence of approximations to a set of basis vectors of the ground-state space (it requires new ideas to keep track of such a set of basis vectors). The construction is more efficient than the approach (Corollary 2.4 and Section 6.2) in \cite{AKLV13}, resulting in an area law for the Renyi entanglement entropy. Finally, efficient MPS representations follow from the decay of the Schmidt coefficients.

\section{Perturbation theory} \label{s3}

Assume without loss of generality that the original 1D Hamiltonian is $H'=\sum_{i=-n}^nH'_i$, where $0\le H'_i\le1$ acts on the spins $i$ and $i+1$. Consider the middle cut. Let $\epsilon_0(\cdot)$ denote the ground-state energy of a Hamiltonian. Define
\begin{equation}
H=H_L+H_{-s}+H_{1-s}+\cdots+H_{s-1}+H_s+H_R
\end{equation}
as\\
(i) $H_L=H'_L-\epsilon_0(H'_L)$ and $H_R=H'_R-\epsilon_0(H'_R)$, where $H'_L:=\sum_{i=-n}^{-s-1}H'_i$ and $H'_R:=\sum_{i=s+1}^{n}H'_i$;\\
(ii) $H_i=H'_i$ for $i=\pm s$;\\
(iii) $H_i=H'_i-\epsilon_0(H'_M)/(2s-1)$ for $1-s\le i\le s-1$, where $H'_M:=\sum_{i=1-s}^{s-1}H'_i$.\\
Hence,\\
(a) $H_L\ge0,~H_R\ge0,$ and $\epsilon_0(H_L)=\epsilon_0(H_R)=0$;\\
(b) $0\le H_i\le1$ for $i=\pm s$;\\
(c) $0\le\sum_{i=1-s}^{s-1}H_i\le2s-1$ and $\epsilon_0(\sum_{i=1-s}^{s-1}H_i)=0$;\\
(d) $H=H'-\epsilon_0(H'_L)-\epsilon_0(H'_M)-\epsilon_0(H'_R)$ so that the (degenerate) ground states and the energy gap are preserved.

Suppose the ground states of $H$ are $f$-fold degenerate, where $f=O(1)$ is assumed to be an absolute constant. Let $0\le\epsilon_0=\epsilon_1=\cdots=\epsilon_{f-1}<\epsilon_f\le\epsilon_{f+1}\le\cdots$ be the lowest energy levels of $H$ with the energy gap $\epsilon:=\epsilon_f-\epsilon_0$. Define
\begin{equation} 
H_L^{\le t}=H_LP^{\le t}_L+t(1-P^{\le t}_L),
\end{equation}
where $P^{\le t}_L$ is the projection onto the subspace spanned by the eigenstates of $H_L$ with eigenvalues at most $t$. $H_R^{\le t}$ is defined analogously. Let
\begin{equation}
H^{(t)}:=H_L^{\le t}+H_{-s}+H_{1-s}+\cdots+H_{s-1}+H_s+H_R^{\le t}\le 2t+2s+1
\end{equation}
be the truncated Hamiltonian with the lowest energy levels $0\le\epsilon'_0\le\epsilon'_1\le\cdots$ and the corresponding (orthonormal) eigenstates $\phi_0^{(t)},\phi_1^{(t)},\ldots$. Note that all states are normalized unless otherwise stated. Define $\epsilon'=\epsilon'_f-\epsilon'_0$ as the energy gap of $H^{(t)}$. Let $B:=H_{-s}+H_s$ be the sum of boundary terms, and $P_t$ be the projection onto the subspace spanned by the eigenstates of $H-B$ with eigenvalues at most $t$ so that
\begin{equation} \label{Pt}
H_LP_t=H_L^{\le t}P_t,H_RP_t=H_R^{\le t}P_t~\Rightarrow~HP_t=H^{(t)}P_t.
\end{equation}

\begin{lemma} \label{l4}
$0\le\epsilon'_0\le\epsilon_0\le2$ and $\epsilon'_f\le\epsilon_f\le[\log_2f]+4=O(1)$.
\end{lemma}

\begin{proof}
Let $\psi_0,\psi_L,\psi_M,\psi_R$ be the ground states of $H,H_L,\sum_{i=1-s}^{s-1}H_i,H_R$, respectively.
\begin{eqnarray}
&&\epsilon_0\le\langle\psi_L\psi_M\psi_R,H\psi_L\psi_M\psi_R\rangle\nonumber\\
&&=\langle\psi_L,H_L\psi_L\rangle+\left\langle\psi_M,\sum_{i=1-s}^{s-1}H_i\psi_M\right\rangle+\langle\psi_R,H_R\psi_R\rangle+\langle\psi_L\psi_M\psi_R,B\psi_L\psi_M\psi_R\rangle\le\|B\|\le2.
\end{eqnarray}
Let $f'=[\log_2f]+1$ and $\phi_R$ be the ground state of $\sum_{i=f'-s+1}^sH_i+H_R$. For any state $\phi_M$ of the spins $1-s,2-s,\cdots,f'-s$,
\begin{eqnarray}
&&\langle\psi_L\phi_M\phi_R,H\psi_L\phi_M\phi_R\rangle\nonumber\\
&&=\langle\psi_L,H_L\psi_L\rangle+\left\langle\psi_L\phi_M\phi_R,\sum_{i=-s}^{f'-s}H_i\psi_L\phi_M\phi_R\right\rangle+\left\langle\phi_R,\left(\sum_{i=f'+1-s}^sH_i+H_R\right)\phi_R\right\rangle\nonumber\\
&&\le\langle\psi,H_L\psi\rangle+\left\langle\psi,\sum_{i=-s}^{f'-s}H_i\psi\right\rangle+f'+1+\left\langle\psi,\left(\sum_{i=f'+1-s}^sH_i+H_R\right)\psi\right\rangle\nonumber\\
&&\le\langle\psi,H\psi\rangle+f'+1=\epsilon_0+f'+1\le f'+3~\Rightarrow\epsilon_f\le f'+3=[\log_2f]+4.
\end{eqnarray}
\end{proof}

Let $\phi^{(r)}$ be an eigenstate of $H^{(r)}$ with eigenvalue $\epsilon^{(r)}$.

\begin{lemma} \label{l5}
For $r,t>\epsilon^{(r)}$,
\begin{equation}
\|(1-P_t)\phi^{(r)}\|^2\le|\langle\phi^{(r)},(1-P_t)BP_t\phi^{(r)}\rangle|/(\min\{r,t\}-\epsilon^{(r)}).
\end{equation}
\end{lemma}

\begin{proof}
It follows from
\begin{eqnarray}
&&\epsilon^{(r)}=\langle\phi^{(r)},H^{(r)}\phi^{(r)}\rangle\nonumber\\
&&=\langle\phi^{(r)},(1-P_t)H^{(r)}(1- P_t)\phi^{(r)}\rangle+\langle\phi^{(r)},P_tH^{(r)}\phi^{(r)}\rangle+\langle\phi^{(r)},(1-P_t)H^{(r)}P_t\phi^{(r)}\rangle\nonumber\\
&&\ge\langle\phi^{(r)},(1-P_t)(H^{(r)}-B)(1- P_t)\phi^{(r)}\rangle+\epsilon^{(r)}\|P_t\phi^{(r)}\|^2\nonumber\\
&&+\langle\phi^{(r)},(1-P_t)(H^{(r)}-B)P_t\phi^{(r)}\rangle+\langle\phi^{(r)},(1-P_t)BP_t\phi^{(r)}\rangle\nonumber\\
&&\ge\min\{r,t\}\|(1- P_t)\phi^{(r)}\|^2+\epsilon^{(r)}(1-\|(1-P_t)\phi^{(r)}\|^2)-|\langle\phi^{(r)},(1-P_t)BP_t\phi^{(r)}\rangle|.
\end{eqnarray}
\end{proof}

Suppose $\epsilon^{(r)}=O(1)$ and $r\ge\epsilon^{(r)}+100=O(1)$.

\begin{lemma} \label{l6}
\begin{equation}
\|(1-P_t)\phi^{(r)}\|\le2^{-\Omega(t)}.
\end{equation}
\end{lemma}

\begin{proof}
Let $t_0=\epsilon^{(r)}+100=O(1)$. We show that there exists $c=O(1)$ such that 
\begin{equation} \label{ind}
\|(1-P_{t_i})\phi^{(r)}\|\le2^{-i}
\end{equation}
for $t_i=t_0+ci$. The proof is an induction on $i$ with fixed $r$. Clearly, (\ref{ind}) holds for $i=0$. Suppose (\ref{ind}) holds for $i=0,1,\ldots,j-1$. Let $P_{t_{-1}}=0$ for notational convenience. Lemma \ref{l5} implies
\begin{eqnarray}
&&\|(1-P_{t_j})\phi^{(r)}\|^2\le|\langle\phi^{(r)},(1-P_{t_j})BP_{t_j}\phi^{(r)}\rangle|/(\min\{r,t_j\}-\epsilon^{(r)})\nonumber\\
&&\le\left|\left\langle\phi^{(r)},(1-P_{t_j})B\sum_{i=0}^j(P_{t_i}-P_{t_{i-1}})\phi^{(r)}\right\rangle\right|/100\nonumber\\
&&\le\|(1-P_{t_j})\phi^{(r)}\|\sum_{i=0}^j\|(1-P_{t_j})B(P_{t_i}-P_{t_{i-1}})\|\|(P_{t_i}-P_{t_{i-1}})\phi^{(r)}\|/100\nonumber\\
&&\Rightarrow\|(1-P_{t_j})\phi^{(r)}\|\le\sum_{i=0}^j\|(1-P_{t_j})BP_{t_i}\|\|(1-P_{t_{i-1}})\phi^{(r)}\|/100\le\sum_{i=0}^je^{(t_i-t_j)/8}2^{-i}/10,
\end{eqnarray}
where we have used the induction hypothesis (\ref{ind}) and the inequality $\|(1-P_{t_j})BP_{t_i}\|\le4e^{(t_i-t_j)/8}$ (Lemma 6.6(2) in \cite{AKLV13}). Hence (\ref{ind}) holds for $i=j$ by setting $c=16\ln2$.
\end{proof}

Let $\Phi^{(t)}:=P_t\phi^{(t)}/\|P_t\phi^{(t)}\|$.

\begin{lemma} \label{l7}
\begin{equation}
\langle\Phi^{(t)},H\Phi^{(t)}\rangle\le\epsilon^{(t)}+2^{-\Omega(t)}.
\end{equation}
\end{lemma}

\begin{proof}
(\ref{Pt}) implies
\begin{eqnarray}
&&\epsilon^{(t)}=\langle\phi^{(t)},H^{(t)}\phi^{(t)}\rangle\nonumber\\
&&\ge\langle\phi^{(t)},P_tH^{(t)}P_t\phi^{(t)}\rangle+\langle\phi^{(t)},P_tH^{(t)}(1-P_t)\phi^{(t)}\rangle+\langle\phi^{(t)},(1-P_t)H^{(t)}P_t\phi^{(t)}\rangle\nonumber\\
&&=\langle\phi^{(t)},P_tHP_t\phi^{(t)}\rangle+\langle\phi^{(t)},P_tB(1-P_t)\phi^{(t)}\rangle+\langle\phi^{(t)},(1-P_t)BP_t\phi^{(t)}\rangle\nonumber\\
&&\ge\langle\phi^{(t)},P_tHP_t\phi^{(t)}\rangle-2\|BP_t\phi^{(t)}\|\cdot\|(1-P_t)\phi^{(t)}\|\ge\langle\phi^{(t)},P_tHP_t\phi^{(t)}\rangle-2^{-\Omega(t)}\nonumber\\
&&\Rightarrow\langle\Phi^{(t)},H\Phi^{(t)}\rangle\le(\epsilon^{(t)}+2^{-\Omega(t)})/\|P_t\phi^{(t)}\|^2=(\epsilon^{(t)}+2^{-\Omega(t)})/(1-2^{-\Omega(t)})=\epsilon^{(t)}+2^{-\Omega(t)}.
\end{eqnarray}
\end{proof}

\begin{remark} 
Suppose $r\ge t$. A very minor modification of the proof implies
\begin{equation}
\langle\Phi^{(r),t},H\Phi^{(r),t}\rangle\le\epsilon^{(r)}+2^{-\Omega(t)}~\mathrm{for}~\Phi^{(r),t}:=P_t\phi^{(r)}/\|P_t\phi^{(r)}\|.
\end{equation}
\end{remark}

Since the proofs of Lemmas \ref{l4}, \ref{l5}, \ref{l6}, \ref{l7} do not require an energy gap, these lemmas also hold in gapless systems. Let $G$ be the ground-state space of $H$.

\begin{lemma} \label{l3}
For any state $\psi$ with $\langle\psi,H\psi\rangle\le\epsilon_0+\varepsilon$, there exists a state $\psi_g\in G$ such that
\begin{equation}
\|\psi-\psi_g\|^2\le2\varepsilon/\epsilon.
\end{equation}
\end{lemma}

\begin{proof}
The state $\psi$ can be decomposed as
\begin{equation}
\psi=c_g\psi_g+c_e\psi_e,~c_g,c_e\ge0,~c_g^2+c_e^2=1,
\end{equation}
where $\psi_g\in G$ and $\psi_e\perp G$. Then,
\begin{eqnarray}
c_g^2\epsilon_0+c_e^2\epsilon_f\le\langle\psi,H\psi\rangle\le\epsilon_0+\varepsilon\Rightarrow c_e^2\le\varepsilon/\epsilon\Rightarrow\|\psi-\psi_g\|^2=2-2c_g\le2\varepsilon/\epsilon.
\end{eqnarray}
\end{proof}

\begin{theorem} \label{t2}
For $t\ge O(\log\epsilon^{-1})$,\\
(a) $0\le\epsilon_0-\epsilon'_{f-1}\le\epsilon_0-\epsilon'_{f-2}\le\cdots\le\epsilon_0-\epsilon'_0\le2^{-\Omega(t)}$;\\
(b) there exists $\psi_i^{(t)}\in G$ such that $\|\psi_i^{(t)}-\phi_i^{(t)}\|^2\le2^{-\Omega(t)}$ for $i=0,1,\ldots,f-1$;\\
(c) $\epsilon'\ge\epsilon/10$.
\end{theorem}

\begin{proof}
Lemma \ref{l7} implies
\begin{eqnarray}
&&\epsilon'_0\le\epsilon'_1\le\cdots\le\epsilon'_{f-1}\le\epsilon_0\le\langle\Phi^{(t)}_0,H\Phi^{(t)}_0\rangle\le\epsilon'_0+2^{-\Omega(t)},\label{E0}\\
&&\langle\Phi^{(t)}_f,H\Phi^{(t)}_f\rangle\le\epsilon'_f+2^{-\Omega(t)}=\epsilon'_0+\epsilon'+2^{-\Omega(t)}\le\epsilon_0+\epsilon'+2^{-\Omega(t)}.
\end{eqnarray}
(a) follows from (\ref{E0}). Using Lemma \ref{l3}, there exists $\psi^{(t)}_0,\psi^{(t)}_1,\ldots,\psi^{(t)}_f\in G$ such that
\begin{equation} \label{E1}
\|\Phi^{(t)}_i-\psi^{(t)}_i\|^2\le2^{-\Omega(t)}/\epsilon=2^{-\Omega(t)+\log\epsilon^{-1}}
\end{equation}
for $i=0,1,\ldots,f-1$ and
\begin{equation} \label{E2}
\|\Phi^{(t)}_f-\psi^{(t)}_f\|^2\le\epsilon'/\epsilon+2^{-\Omega(t)}/\epsilon.
\end{equation}
Lemma \ref{l6} implies
\begin{equation} \label{E3}
\|\phi^{(t)}_i-\Phi^{(t)}_i\|^2\le2^{-\Omega(t)}.
\end{equation}
(b) follows from (\ref{E1}), (\ref{E3}) as $t\ge O(\log\epsilon^{-1})$. (c) follows from (\ref{E1}), (\ref{E2}), (\ref{E3}), because $\phi^{(t)}_0,\phi^{(t)}_1,\ldots,\phi^{(t)}_f$ are pairwise orthogonal while $\psi^{(t)}_0,\psi^{(t)}_1,\ldots,\psi^{(t)}_f$ are linearly dependent.
\end{proof}
 
\section{Approximate ground-space projection} \label{s2}

Recall that $H^{(t)}$ is the truncated Hamiltonian with the lowest energy levels $0\le\epsilon'_0\le\epsilon'_1\le\cdots$ and the corresponding (orthonormal) eigenstates $\phi_0^{(t)},\phi_1^{(t)},\ldots$. Theorem \ref{t2} implies that the lowest $f$ energy levels are nearly degenerate: $\epsilon'_0\approx\epsilon'_{f-1}$, and $\epsilon'=\epsilon'_f-\epsilon'_0$ is the energy gap. Let $G':=\mathrm{span}\{\phi_i^{(t)}|i=0,1,\ldots,f-1\}$ be the ground-state space of $H^{(t)}$. Let $R(\psi)$ denote the Schmidt rank of a state $\psi$ across the middle cut.

\begin{definition} [Approximate ground-space projection (AGSP) \cite{ALV12}]
A linear operator $A$ is a $(\Delta,D)$-AGSP if\\
(i) $A\psi=\psi$ for $\forall\psi\in G'$;\\
(ii) $A\psi\perp G'$ and $\|A\psi\|^2\le\Delta$ for $\forall\psi\perp G'$;\\
(iii) $R(A\psi)\le DR(\psi)$ for $\forall\psi$.
\end{definition}

Let $\epsilon'_\infty:=2s+2t+1$ be an upper bound on the maximum eigenvalue of $H^{(t)}$.

\begin{lemma} \label{l2}
Suppose $l^2(\epsilon'_{f-1}-\epsilon'_0)/(\epsilon'_\infty-\epsilon'_f)\le1/10$. Then there exists a polynomial $C_l$ of degree $fl$ such that\\
(i) $C_l(\epsilon'_0)=C(\epsilon'_1)=\cdots=C(\epsilon'_{f-1})=1$;\\
(ii) $C_l^2(x)\le2^{2f+4}e^{-4l\sqrt{\epsilon'/\epsilon'_\infty}}$ for $\epsilon'_f\le x\le\epsilon'_\infty$.
\end{lemma}

\begin{proof}
The Chebyshev polynomial of the first kind of degree $l$ is defined as
\begin{equation}
T_l(x)=\cos(l\arccos x)=\cosh(ly),~y:=\mathrm{arccosh}x.
\end{equation}
By definition, $|T_l(x)|\le 1$ for $|x|\le1$. For $x\ge1$, $T_l(x)$ is monotonically increasing function of $x$, and
\begin{equation}
T_l(x)\ge e^{ly}/2\ge e^{2l\tanh(y/2)}/2=e^{2l\sqrt{(x-1)/(x+1)}}/2,~\frac{T'_l(x)}{T_l(x)}=\frac{l\tanh(ly)}{\sinh y}\le\frac{l(l y)}{y}=l^2.
\end{equation}
Let $g(x):=(\epsilon'_\infty+\epsilon'_f-2x)/(\epsilon'_\infty-\epsilon'_f)$ such that $g(\epsilon'_\infty)=-1$ and $g(\epsilon'_f)=1$. Define $S_l(x)=T_l(g(x))$ as a polynomial of degree $l$. Clearly, $|S_l(x)|\le1$ for $\epsilon'_f\le x\le\epsilon'_\infty$ and
\begin{equation}
S_l(\epsilon'_0)=T_l(g(\epsilon'_0))\ge e^{2l\sqrt{(g(\epsilon'_0)-1)/(g(\epsilon'_0)+1)}}/2\ge e^{2\l\sqrt{\epsilon'/\epsilon'_\infty}}/2.
\end{equation}
There exists $\epsilon'_0\le\xi\le\epsilon'_{f-1}$ such that
\begin{eqnarray}
&&S_l(\epsilon'_{f-1})=S_l(\epsilon'_0)+(\epsilon'_{f-1}-\epsilon'_0)S'_l(\xi)\ge S_l(\epsilon'_0)(1+(\epsilon'_{f-1}-\epsilon'_0)T'_l(g(\xi))g'(\xi)/T_l(g(\xi)))\nonumber\\
&&\Rightarrow S_l(\epsilon'_{f-1})/S_l(\epsilon'_0)\ge1-2l^2(\epsilon'_{f-1}-\epsilon'_0)/(\epsilon'_\infty-\epsilon'_f)\ge4/5.
\end{eqnarray}
Assume without loss of generality that $\epsilon'_0,\epsilon'_1,\ldots,\epsilon'_{f-1}$ are pairwise distinct. Let $L(x)=\sum_{i=1}^fa_ix^i$ be the Lagrange interpolation polynomial of degree $f$ such that $L(0)=0$ and $L(S_l(\epsilon'_0))=L(S_l(\epsilon'_1))=\cdots=L(S_l(\epsilon'_{f-1}))=S_l(\epsilon'_0)$. For each $i=1,2,\ldots,f-1$, there exists $S_l(\epsilon'_{i-1})>\xi_i>S_l(\epsilon'_i)$ such that $L'(\xi_i)=0$. Then,
\begin{equation}
L'(x)=a_1\prod_{i=1}^{f-1}(1-x/\xi_i).
\end{equation}
Clearly, $a_1>0$ and $L'(x)>0$ for $x<S_l(\epsilon'_{f-1})$. Hence,
\begin{eqnarray}
&&S_l(\epsilon'_0)=L(S_l(\epsilon'_{f-1}))=\int_0^{S_l(\epsilon'_{f-1})}L'(x)\mathrm{d}x\ge a_1\int_0^{S_l(\epsilon'_{f-1})}(1-x/S_l(\epsilon'_{f-1}))^{f-1}\mathrm{d}x\nonumber\\
&&=a_1S_l(\epsilon'_{f-1})/f\Rightarrow a_1\le5f/4.
\end{eqnarray}
For $|x|\le1$,
\begin{equation}
\xi_1>\xi_2>\cdots>\xi_{f-1}>S_l(\epsilon'_f)=1\Rightarrow|L'(x)|\le a_1(1+|x|)^{f-1}\Rightarrow|L(x)|\le2^{f+1}.
\end{equation}
Finally, $C_l(x):=L(S_l(x))/S_l(\epsilon'_0)$ is a polynomial of degree $fl$.
\end{proof}

\begin{lemma} [Lemma 4.2 in \cite{AKLV13}] \label{l8}
For any polynomial $p_l$ of degree $l\le s^2$ and any $t,\psi$,
\begin{equation}
R(p_l(H^{(t)})\psi)\le l^{O(\sqrt l)}R(\psi).
\end{equation}
\end{lemma}

Let $l=s^2/f$ and $t=\Omega(s)$. The assumption 
\begin{equation}
1/10\ge l^2(\epsilon'_{f-1}-\epsilon'_0)/(\epsilon'_\infty-\epsilon'_f)=O(s^42^{-\Omega(t)}/(s+t))=O(s^32^{-\Omega(s)})
\end{equation}
is satisfied with sufficiently large $s>O(1)$. Lemmas \ref{l2}, \ref{l8} imply a $(\Delta,D)$-AGSP $A=C_l(H^{(t)})$ for $H^{(t)}$ with
\begin{equation}
\Delta=2^{2f+4}e^{-4l\sqrt{\epsilon'/\epsilon'_\infty}}=2^{-\Omega(s^2\sqrt{\epsilon/t})},~D=(s^2)^{O(\sqrt{s^2})}=s^{O(s)}.
\end{equation}
In particular, the condition
\begin{equation}
1/100\ge\Delta D^2=2^{-\Omega(s^2\sqrt{\epsilon/t})}s^{O(s)}\Rightarrow1/100\ge\Delta D
\end{equation}
can be satisfied by fixing $t=t_0=\Theta(s_0)$ and $s=s_0=\tilde O(\epsilon^{-1})$ so that $\Delta=2^{-\tilde\Omega(\epsilon^{-1})}$ and $D=2^{\tilde O(\epsilon^{-1})}$.

\section{Area law} \label{s4}

Hereafter $f=2$ is assumed for ease of presentation. It should be clear that a very minor modification of the proof works for any $f=O(1)$. Suppose $s=s_0$ and $t=t_0$ as given above so that $A$ is a $(\Delta,D)$-AGSP for $H^{(t_0)}$ with $\Delta D^2\le1/100$. Recall that $\phi_0^{(t_0)},\phi_1^{(t_0)}$ are the lowest two eigenstates and $G'=\mathrm{span}\{\phi_0^{(t_0)},\phi_1^{(t_0)}\}$ is the ground-state space of $H^{(t_0)}$.

\begin{lemma} \label{ls}
There exist $\varphi_0,\varphi_1\in G'$ and $\psi_0,\psi'_0$ such that
(i) $\varphi_0\perp\varphi_1$;
(ii) $|\langle\varphi_0,\psi_0\rangle|^2\ge24/25$;
(iii) $R(\psi_0)=2^{\tilde O(\epsilon^{-1})}$;
(iv) $|\langle\varphi_1,\psi'_0\rangle|^2\ge24/25$;
(v) $R(\psi'_0)=2^{\tilde O(\epsilon^{-1})}$.
\end{lemma}

\begin{proof}
Let $P'$ be the projection onto $G'$. Consider
\begin{equation} \label{opt}
\max_{R(\psi)=1}\|P'\psi\|^2.
\end{equation}
As the set $\{\psi|R(\psi)=1\}$ of product states is compact, the optimal state exists and is still denoted by $\psi$. This state and $\phi:=A\psi$ can be decomposed as
\begin{equation}
\psi=c_g\psi_g+c_e\psi_e,~\phi=c'_g\phi_g+c'_e\phi_e,
\end{equation}
where $\psi_g,\phi_g\in G'$ and $\psi_e,\phi_e\perp G'$. The definition of AGSP implies
\begin{equation}
c_g=c'_g,~\psi_g=\phi_g,~|c'_e|^2\le\Delta,~R(\phi)\le D.
\end{equation}
The Schmidt decomposition of the unnormalized state $\phi$ implies
\begin{equation}
\phi=\sum_{i=1}^{R(\phi)}\lambda_iL_i\otimes R_i\Rightarrow\sum_{i=1}^{R(\phi)}\lambda_i^2=\|\phi\|^2=|c'_g|^2+|c'_e|^2\le|c_g|^2+\Delta.
\end{equation}
Since $|c_g|^2$ is the optimal value in (\ref{opt}),
\begin{eqnarray}
&&|c_g|=|\langle\psi_g,\phi\rangle|\le\sum_{i=1}^{R(\phi)}\lambda_i|\langle\psi_g,L_i\otimes R_i\rangle|\le\sum_{i=1}^{R(\phi)}\lambda_i\|P'L_i\otimes R_i\|\le|c_g|\sum_{i=1}^{R(\phi)}\lambda_i\Rightarrow1\le\left(\sum_{i=1}^{R(\phi)}\lambda_i\right)^2\nonumber\\
&&\le R(\phi)\sum_{i=1}^{R(\phi)}\lambda_i^2\le D(|c_g|^2+\Delta)\le D|c_g|^2+1/100\Rightarrow|c_g|^2\ge99D^{-1}/100\ge99\Delta.
\end{eqnarray}
Applying the AGSP twice, the state $\psi_0:=A^2\psi/\|A^2\psi\|$ satisfies
\begin{equation} \label{r1}
\|P'\psi_0\|^2\ge1-\Delta/50,~R(\psi_0)=D^2=2^{\tilde O(\epsilon^{-1})}.
\end{equation}
Define $\varphi_0=P'\psi_0/\|P'\psi_0\|\in G'$ and $\varphi_1\in G'$ such that $\varphi_0\perp\varphi_1$. Clearly,
\begin{equation}
|\langle\varphi_0,\psi_0\rangle|^2\ge1-\Delta/50,~\langle\varphi_1,\psi_0\rangle=0,~|\langle\varphi_e,\psi_0\rangle|^2\le\Delta/50~\mathrm{for}~\forall~\varphi_e\perp G'.
\end{equation}

Consider
\begin{equation}
\max_{R(\psi')=1}|\langle\varphi_1,\psi'\rangle|^2.
\end{equation}
As the set $\{\psi'|R(\psi')=1\}$ of product states is compact, the optimal state exists and is still denoted by $\psi'$. This state and $\phi':=A\psi'-\langle\psi_0,\psi'\rangle\psi_0$ can be decomposed as
\begin{equation}
\psi'=c_0\varphi_0+c_1\varphi_1+c_e\varphi_e,~\phi'=c_1\varphi_1+c_r\varphi_r,
\end{equation}
where $\varphi_e\perp G'$ and $\varphi_r\perp\varphi_1$. Specifically,
\begin{eqnarray}
&&c_r\varphi_r=c_0(A\varphi_0-\langle\psi_0,\varphi_0\rangle\psi_0)-c_1\langle\psi_0,\varphi_1\rangle\psi_0+c_e(A\varphi_e-\langle\psi_0,\varphi_e\rangle\psi_0)\nonumber\\
&&\Rightarrow|c_r|\le0.2|c_0|\sqrt\Delta+1.2|c_e|\sqrt\Delta\le1.4\sqrt\Delta~\mathrm{and}~R(\phi')\le D+R(\psi_0)\le D+D^2\le2D^2.
\end{eqnarray}
The Schmidt decomposition of the unnormalized state $\phi'$ implies
\begin{equation}
\phi'=\sum_{i=1}^{R(\phi')}\lambda'_iL'_i\otimes R'_i\Rightarrow\sum_{i=1}^{R(\phi')}\lambda_i'^2=\|\phi'\|^2=|c_1|^2+|c_r|^2\le|c_1|^2+2\Delta.
\end{equation}
Since $\psi'$ is the optimal state,
\begin{eqnarray}
&&|c_1|=|\langle\varphi_1,\phi'\rangle|\le\sum_{i=1}^{R(\phi')}\lambda'_i|\langle\varphi_1,L'_i\otimes R'_i\rangle|\le\sum_{i=1}^{R(\phi')}\lambda'_i|\langle\varphi_1,\psi'\rangle|=|c_1|\sum_{i=1}^{R(\phi')}\lambda'_i\Rightarrow1\le\left(\sum_{i=1}^{R(\phi')}\lambda'_i\right)^2\nonumber\\
&&\le R(\phi')\sum_{i=1}^{R(\phi')}\lambda_i'^2\le2D^2(|c_1|^2+2\Delta)\le2D^2|c_1|^2+1/25\Rightarrow|c_1|^2\ge12D^{-2}/25\ge48\Delta.
\end{eqnarray}
Hence $\psi'_0=\phi'/\|\phi'\|$ is a state with $R(\psi'_0)=R(\phi')\le2D^2=2^{\tilde O(\epsilon^{-1})}$ and $|\langle\varphi_1,\psi'_0\rangle|^2\ge24/25$.
\end{proof}

Recall that $G$ is the ground-state space of $H$.

\begin{lemma}
For any $\Psi\in G$, there is a sequence of approximations $\{\Psi_i\}$ such that\\
(a) $|\langle\Psi_i,\Psi\rangle|\ge1-2^{-\Omega(i)}$;\\
(b) $R_i:=R(\Psi_i)=2^{\tilde O(\epsilon^{-1}+\epsilon^{-1/4}i^{3/4})}$.
\end{lemma}

\begin{proof}
Let $t_i=t_0+i$. Theorem \ref{t2}(b) is a quantitative statement that $G$ and $\mathrm{span}\{\phi_0^{(t_i)},\phi_1^{(t_i)}\}$ are exponentially close. In particular, setting $t_0$ to be a sufficiently large constant implies that $G'$ and $\mathrm{span}\{\phi_0^{(t_i)},\phi_1^{(t_i)}\}$ are close up to a small constant. Hence Lemma \ref{ls}(ii) implies
\begin{equation} \label{almost}
|\langle\phi_0^{(t_i)},\psi_0\rangle|^2+|\langle\phi_1^{(t_i)},\psi_0\rangle|^2\ge9/10.
\end{equation}
Let $l_i=s_i^2/2=\Theta(\sqrt{t_i^3/\epsilon})=O(t_i^2)$ such that the assumption
\begin{equation}
1/10\ge l_i^2(\epsilon'_1-\epsilon'_0)/(\epsilon'_\infty-\epsilon'_2)=O(s_i^32^{-\Omega(s_i)})
\end{equation}
is satisfied with sufficiently large $s_i>O(1)$. Lemmas \ref{l2}, \ref{l8} imply a $(\Delta_i,D_i)$-AGSP $A_i=C_{l_i}(H^{(t_i)})$ for $H^{(t_i)}$ with
\begin{equation}
\Delta_i=2^{-\Omega(s_i^2\sqrt{\epsilon/t_i})}=2^{-\Omega(t_i)},~D_i=s_i^{O(s_i)}=2^{\tilde O(\epsilon^{-1/4}t_i^{3/4})}.
\end{equation}
Hence the sequence of operators $\{A_i\}_{i=1}^{+\infty}$ converges exponentially due to Theorem \ref{t2}(b). Clearly, $A_\infty:=\lim_{i\rightarrow+\infty}A_i$ is just the projection onto $G$. Let $\psi_i:=A_i\psi_0/\|A_i\psi_0\|$ with $\psi_\infty\in G$ such that\\
\begin{equation}
R(\psi_i)\le R(\psi_0)D_i\le2^{\tilde O(\epsilon^{-1}+\epsilon^{-1/4}t_i^{3/4})},~|\langle\psi_i,\psi_\infty\rangle|\ge1-2^{-\Omega(t_i)}.
\end{equation}
Similarly, Let $\psi'_i:=A_i\psi'_0/\|A_i\psi'_0\|$ with $\psi'_\infty\in G$ such that\\
\begin{equation}
R(\psi'_i)\le2^{\tilde O(\epsilon^{-1}+\epsilon^{-1/4}t_i^{3/4})},~|\langle\psi'_i,\psi'_\infty\rangle|\ge1-2^{-\Omega(t_i)}.
\end{equation}
(\ref{almost}) with $i=+\infty$ is a quantitative statement that $\psi_0$ is close to $G$, and hence $\psi_0$ and $\psi_\infty$ are close up to a small constant. Since $\psi_0$ and $\varphi_0$ are close up to a small constant, $\psi_\infty$ and $\varphi_0$ are also close. The same arguments imply that $\psi'_\infty$ and $\varphi_1$ are close. Hence, $\psi_\infty$ and $\psi'_\infty$ are almost orthogonal. Any state $\Psi\in G$ can be decomposed as
\begin{equation}
\Psi=c\psi_\infty+c'\psi'_\infty,~|c|=O(1),~|c'|=O(1).
\end{equation}
Then, $\{\Psi_i:=c\psi_i+c'\psi'_i\}_{i=0}^{+\infty}$ is a sequence of approximations to $\Psi$ with (b) $R(\Psi_i)=2^{\tilde O(\epsilon^{-1}+\epsilon^{-1/4}t_i^{3/4})}$. (a) also follows immediately.
\end{proof}

\begin{proof} [Proof of Theorem \ref{t1}]
(a) Let $\Lambda_i$ be the Schmidt coefficients of $\Psi$ across the middle cut. Then,
\begin{equation}
1-p_i:=\sum_{j=1}^{R_i}\Lambda_j^2\ge|\langle\Psi_i,\Psi\rangle|^2\ge1 - 2^{-\Omega(i)}.
\end{equation}
The Renyi entanglement entropy of $\Psi$ is upper bounded by
\begin{eqnarray}
&&\frac{\log\left(R^{1-\alpha}_0+\sum_{i=0}^{+\infty}p_i^\alpha(R_{i+1}-R_i)^{1-\alpha}\right)}{1-\alpha}\le\frac{\log\left(2^{(1-\alpha)\tilde O(\epsilon^{-1})}+\sum_{i=0}^{+\infty}2^{(1-\alpha)\tilde O(\epsilon^{-1}+\epsilon^{-1/4}i^{3/4})-\alpha\Omega(i)}\right)}{1-\alpha}\nonumber\\
&&=\tilde O(\epsilon^{-1})+\frac{\log(O(1)+2^{(1-\alpha)\tilde O((1-\alpha)^3\alpha^{-3}/\epsilon)})}{1-\alpha}=\tilde O(\epsilon^{-1}+(1-\alpha)^3\alpha^{-3}/\epsilon)=\tilde O(\alpha^{-3}\epsilon^{-1}).
\end{eqnarray}

(b) Finally we sketch the proof that $\Psi$ is well approximated by an MPS of small bond dimension. We first express it exactly as an MPS of possibly exponential (in $n$) bond dimension and then truncate the MPS cut by cut. It is shown in \cite{VC06} the error accumulates at most additively: If an inverse polynomial overall error $1/p(n)=1/\mathrm{poly}(n)$ is allowed, it suffices that the error of truncating each cut is $1/(np(n))=1/\mathrm{poly}(n)$. We require that
\begin{equation}
1/\mathrm{poly}(n)=p_i\Rightarrow i=O(\log n),
\end{equation}
and hence the bond dimension is $2^{\tilde O(\epsilon^{-1/4}\log^{3/4}n)}$.
\end{proof}

\section{Notes} \label{notes}

For nondegenerate systems ($f=1$), the upper bound claimed in \cite{AKLV13} on the von Neumann entanglement entropy is $\tilde O(\epsilon^{-1})$. However, the proof in \cite{AKLV13} of this claim appears incomplete. Specifically, in Lemma 6.3 in \cite{AKLV13} $t_0$ should be at least $O(\epsilon_0/\epsilon^2+\epsilon^{-1})$ in order that the robustness theorem (Theorem 6.1 in \cite{AKLV13}) applies to $H^{(t_0)}$, i.e., the robustness theorem does not guarantee that $H^{(t_0)}$ is gapped if $t_0=O(1)$. Then $s=\tilde O(\epsilon^{-1})$ (and $l=s^2$) does not give an AGSP for $H^{(t_0)}$ with $\Delta D\le1/2$, but $s=\tilde O(\epsilon^{-3/2})$ does. A straightforward calculation shows that the upper bound $\tilde O(\epsilon^{-3/2})$ on the von Neumann entanglement entropy follows from the proof in \cite{AKLV13}. Nevertheless, in the present paper I have shown that the claim in \cite{AKLV13} is correct, because Theorem \ref{t2} (as a stronger version of the robustness theorem) only requires $t\ge O(\log\epsilon^{-1})$.

After the appearance of the present paper on arXiv \cite{area1}, Section \ref{s3} (perturbation theory) was extended to higher dimensions \cite{AKL14}. In particular, Theorems 4.2, 4.6 in \cite{AKL14} are generalizations of Lemmas \ref{l6}, \ref{l7}, respectively.

\bibliography{Hua14}

\end{document}